\documentclass[letterpaper, 10 pt, conference]{ieeeconf}  
\IEEEoverridecommandlockouts
\usepackage{balance}
\usepackage{color}
\usepackage{caption}
\usepackage{subcaption}
\usepackage{enumerate}
\usepackage{cite}
\usepackage{empheq}
\usepackage{mathrsfs}

\usepackage{mathtools}

\usepackage{tikz}
\usepackage{circuitikz}

\usepackage{graphicx,amssymb,amstext,amsmath}

\makeatletter
\let\proof\@undefined                        
\let\endproof\@undefined                  
\makeatother
\usepackage{graphicx,amssymb,amstext,amsmath,amsthm}

\usepackage[bookmarks=true]{hyperref}
\usepackage{algorithm,algorithmicx,algpseudocode}
\algnewcommand{\algorithmicgoto}{\textbf{go to}}%
\algnewcommand{\Goto}[1]{\algorithmicgoto~\ref{#1}}%
\algnewcommand{\LineComment}[1]{\Statex \(\triangleright\) #1}
\algnewcommand{\LineCommentN}[1]{\Statex \hspace{1cm}\(\triangleright\) #1}

\usepackage{multirow}
\usepackage{stfloats}

\newcommand{\argmin}{\operatornamewithlimits{arg\ min}}

\newtheorem{thm}{Theorem}

\newtheorem{lem}{Lemma}

\newtheorem{rem}{Remark}
\newtheorem{problem}{Problem}

\let\oldbibliography\thebibliography
\renewcommand{\thebibliography}[1]{%
	\oldbibliography{#1}%
}

\begin{document}
	
	\title{\LARGE \bf 
	Simultaneous Input and State Set-Valued $\mathcal{H}_\infty$-Observers For Linear Parameter-Varying  
	 Systems} 		
	\author{%
		 Mohammad Khajenejad \qquad Sze Zheng Yong \\
		\thanks{
			  M. Khajenejad and S.Z. Yong are with the School for Engineering of Matter, Transport and Energy, Arizona State University, Tempe, AZ, USA (e-mail:  mkhajene@asu.edu, szyong@asu.edu ).}
	}
	
	\maketitle
	\thispagestyle{empty}
	\pagestyle{empty}
	
	\begin{abstract}
		A fixed-order set-valued observer is presented for linear parameter-varying systems with bounded-norm noise and under completely unknown attack signals, 
		which simultaneously finds bounded sets of states and unknown inputs that include the true state and inputs. The proposed observer can be designed using semidefinite programming with LMI constraints and is optimal in the minimum $\mathcal{H}_\infty$-norm sense. We show that the strong detectability of each \emph{constituent} linear time-invariant system is a necessary condition for the existence of such an observer, as well as the boundedness of set-valued estimates. 
		 Furthermore, sufficient conditions are provided for the upper bounds of the estimation errors to converge to steady state values and finally, the results of such a set-valued observer are exhibited through an illustrative example. 

 \end{abstract}
\vspace{-0.15cm}
\section{Introduction}
The security of Cyber-Physical Systems (CPS) is emerging as an extremely critical and important issue. Since physical and software components are deeply intertwined in CPS,
such systems are potentially vulnerable to adversarial attacks, which could be harmful for both the physical systems and their operators. 
	Given that adversarial attackers may behave strategically, there are many potential avenues through which they can cause harm, steal information/power, etc. 
	
	Misleading the system operator by inserting counterfeit data into sensor and actuator signals (false data injection) is among  the most  common and extensive attacks on CPS. Hence, a significant amount of effort has been invested in new designs for estimation and control against false data injection attacks. 
	Due to the nature of the attack signals, it is not justifiable to impose any kind of restrictive assumptions on them (e.g., stochastic with normal distribution or deterministic with bounded norm) i.e., they can be \emph{anything}, so none of the classical Kalman filtering based methods are applicable. Moreover, considering the complicated structure of CPS (several physical, computing and communication components), they can be modeled and represented more realistically by time-varying and nonlinear dynamic systems rather than linear time-invariant ones. Taking all these facts into account, this paper attempts to design a resilient observer for a particular class of linear time-varying systems known as linear parameter-varying systems, which to the best of our knowledge is a novel approach.
\\ \emph{Literature review.} There are several different frameworks for  simultaneous input and state estimation of linear time-varying stochastic systems with unknown inputs, assuming that the noise signals are Gaussian and white. 
The authors in \cite{Gillijns.2007,Gillijns.2007b,Yong.Zhu.ea.CDC15_General,Yong.Zhu.ea.Automatica16}
    apply Kalman filtering inspired recursive filter design approaches (modified versions of unbiased minimum-variance estimation methods), where some 
    additional assumptions are needed to guarantee the stability of the filters, while \cite{lu2016framework} 
    uses a modified double-model adaptive estimation method. 
   However, these Kalman filtering inspired approaches 
      are not applicable for set-membership estimation problems (cf. \cite{YongACC2018} for a comprehensive discussion), as is considered in this paper. 
      
    In the context of attack-resilient estimation, where adversarial signals can be malicious and strategic and thus, bounds on the unknown inputs/attack signals are completely unknown, there have been a number of proposed approaches in the literature for systems with bounded errors (e.g., \cite{pajic2015attack,nakahira2015dynamic,yong2016robust}),
     but all of them only consider point estimates, i.e, the most likely or best single estimate as opposed to a set-valued estimate. Specifically, the work in \cite{pajic2015attack} only computes error bounds for the initial state and \cite{nakahira2015dynamic} assumes zero initial states and does not consider any optimality criteria.  
 The author in \cite{YongACC2018} and references therein discussed the advantages of set-valued observers (when compared to point estimators) in terms of providing hard accuracy bounds, which are important to guarantee safety \cite{blanchini2012convex}. 
In addition, the use of fixed-order set-valued methods can help decrease the complexity of optimal observers\cite{milanese1991optimal}, 
 which grows with time. 
Hence, the work in \cite{YongACC2018} presents 
a fixed-order set-valued observer for linear time-invariant discrete time systems with bounded errors, that simultaneously finds bounded sets of compatible states and unknown inputs that are optimal in the minimum $\mathcal{H}_\infty$-norm sense, i.e., with minimum average power amplification, 
 which we aim to generalize in this paper for linear parameter-varying systems. 
\\  \emph{Contributions.} 
We propose a novel fixed-order set-valued observer for linear parameter-varying systems with unknown input and bounded noise signals  that simultaneously finds bounded sets of states and unknown inputs that contain the true state and unknown input and are compatible/consistent with the measurement outputs. 
Specifically, we consider linear parameter-varying system dynamics that can be presented as a convex combination of linear time-invariant \emph{constituent} dynamics. 
In addition, we provide necessary conditions for the boundedness of the set-valued estimates. We further prove the optimality of the filter in the minimum $\mathcal{H}_\infty$-norm sense, i.e., minimum average power amplification, by converting the corresponding problem into a tractable formulation using semi-definite programming with LMI constraints  that is readily implementable using off-the-shelf optimization solvers. 
 We also show 
   that strong detectability of each constituent system is a necessary condition for the existence of such an $\mathcal{H}_\infty$-observer.
    Then, we provide some sufficient conditions for the convergence of upper bounds of the state and input estimation errors to steady state  and for obtaining these steady state bounds. Finally, we demonstrate the effectiveness of our proposed set-valued observer through an illustrative example.


\emph{Notation.} 
$\mathbb{R}^n$ denotes the $n$-dimensional Euclidean space 
and $\mathbb{N}$ nonnegative integers. For a vector $v \in \mathbb{R}^n$ and a matrix $M \in \mathbb{R}^{p \times q}$, $\|v\|\triangleq \sqrt{v^\top v}$ and $\|M\|$ denote their  (induced) 2-norm. 
Moreover, the transpose, inverse, 
Moore-Penrose pseudoinverse 
and rank of $M$ are given by $M^\top$, $M^{-1}$, $M^\dagger$ 
and ${\rm rk}(M)$. For a symmetric matrix $S$, $S \succ 0$ ($S \succeq 0$) 
is positive (semi-)definite.

\section{Problem Statement} \label{sec:Problem}
\noindent\textbf{\emph{System Assumptions.}} 
Consider the following linear parameter-varying discrete-time bounded-error system: 
\begin{align} \label{eq:mainSystem}
\begin{split}
x_{k+1}&=\textstyle\sum_{i=1}^N \lambda_{i,k}({A}^{i} x_{k}+{B}^{i} u_k   + {w}^{i}_k)+G d_k,\\
y_k&= {C} x_k +\textstyle\sum_{i=1}^N \lambda_{i,k}({D}^{i} u_k  + {v}^{i}_{k})+ H d_k, \end{split}
\end{align}
where $\lambda_{i,k}$ is known and satisfies $ 0 \leq \lambda_{i,k} \leq1 , \sum_{i=1}^N \lambda_{i,k} =1, \forall k $.
$x_k \in \mathbb{R}^n$ is the state vector at time $k \in \mathbb{N}$, 
$u_k \in \mathbb{R}^m$ is a known input vector, $d_k \in \mathbb{R}^p$ is an unknown input vector, and $y_k \in \mathbb{R}^l$ is the measurement vector. The process noise ${w}^{i}_{k} \in \mathbb{R}^n$ and the measurement noise ${v}^{i}_{k} \in \mathbb{R}^l$ are assumed to be bounded and $\ell_\infty$ sequences, with $\|{w}^{i}_{k}\| \leq \eta_w$ and $\|{v}^{i}_{k}\|\leq \eta_v$.We also assume an estimate $\hat{x}_0$ of the initial state $x_0$ is available, where $\|\hat{x}_0-x_0\|\leq \delta_0^x$. The matrices ${A}^{i}$, ${B}^{i}$, ${C}$, ${D}^{i}$, $G$ and $H$ are known for $ i \in\left\{1,2,\dots ,N\right\} $  
and of appropriate dimensions, where $G$ and $H$ are matrices that 
encode the \emph{locations} through which the unknown input or attack signal can affect the system dynamics and measurements and $ N $ is the number of \emph{constituent} systems.
 Note that no assumption is made on $H$ to be either the zero matrix (no direct feedthrough), or to have full column rank when there is direct feedthrough. Without loss of generality, we assume that ${\rm rk}[G^\top\; H^\top ]=p$, $n \geq l \geq 1$, $l \geq p \geq 0$, $m \geq 0$
 and each $({A}^{i},{B}^{i},{C},{D}^{i},G,H),i \in\left\{1,2,\dots ,N\right\}$ represents a  linear time-invariant constituent system:
\begin{align} \label{eq:indSystem}
\begin{array}{ll}
x^i_{k+1}&= {A}^{i} x^i_k + {B}^{i} u_k + G d_k + {w}^{i}_{k},\\
y^i_k&= {C} x_k + {D}^{i} u_k + H d_k + {v}^{i}_{k}. \end{array}
\end{align}

\noindent \textbf{\emph{Unknown Input (or Attack) Signal Assumptions.}} 
The unknown inputs $d_k$ are not constrained to be a signal of any type (random or strategic) nor to follow any model, thus no prior `useful' knowledge of the dynamics of $d_k$ is available (independent of $\{d_\ell\}$ $\forall k\neq \ell$, $\{w_\ell\}$ and $\{v_\ell\}$ $ \forall  \ell$). We also do not assume that $d_k$ is bounded or has known bounds and thus, $d_k$ is suitable for representing adversarial 
attack signals.

The simultaneous input and state set-valued observer design problem 
can be stated as follows:
\begin{problem}
 Given a linear parameter-varying discrete-time bounded-error system with unknown inputs \eqref{eq:mainSystem}
, design an optimal and stable filter that simultaneously finds bounded sets 
of compatible states and unknown inputs in the minimum $\mathcal{H}_\infty$-norm sense, i.e., with minimum average power amplification. \end{problem}

\section{Preliminary Material}
\subsection{System Transformation}
In order to decouple the output equation into two components, first a transformation is carried out for each of the constituent subsystems, one with a full rank direct feedthrough matrix and the other without direct feedthrough. Note that this similarity transformation is similar to the one in \cite{YongACC2018} 
and is not the same as the one in \cite{Yong.Zhu.ea.Automatica16}, which is no longer applicable as it was based on the noise error covariance.

Let $p_{H}\triangleq {\rm rk} (H)$. Using singular value decomposition, we rewrite the direct feedthrough matrix $H$  as
$H= \begin{bmatrix}U_{1}& U_{2} \end{bmatrix} \begin{bmatrix} \Sigma & 0 \\ 0 & 0 \end{bmatrix} \begin{bmatrix} V_{1}^{\, \top} \\ V_{2}^{\, \top} \end{bmatrix}$, 
where $\Sigma \in \mathbb{R}^{p_{H} \times p_{H}}$ is a diagonal matrix of full rank, $U_{1} \in \mathbb{R}^{l \times p_{H}}$, $U_{2} \in \mathbb{R}^{l \times (l-p_{H})}$, $V_{1} \in \mathbb{R}^{p \times p_{H}}$ and $V_{2} \in \mathbb{R}^{p \times (p-p_{H})}$, while $U\triangleq \begin{bmatrix} U_{1} & U_{2} \end{bmatrix}$ and $V\triangleq \begin{bmatrix} V_{1} & V_{2} \end{bmatrix}$ are unitary matrices. 
When there is no direct feedthrough, $\Sigma$, $U_{1}$ and $V_{1}$ are empty matrices\footnote{\ Based on the convention that the inverse of an empty matrix is an empty matrix and the assumption that operations with empty matrices are possible.}, 
and $U_{2}$ and $V_{2}$ are arbitrary unitary matrices.

Then, 
we decouple the unknown input into two orthogonal components: 
\begin{align}
d_{1,k}=V_{1}^\top d_k, \quad
d_{2,k}=V_{2}^\top d_k.
\end{align}
Considering that $V$ is unitary, $d_k =V_{1} d_{1,k}+V_{2} d_{2,k}$ and  we can represent the system \eqref{eq:mainSystem}  as:
\begin{align}
\nonumber x_{k+1}
&= \textstyle\sum_{i=1}^N \lambda_{i,k} ({A}^{i} x_k + {B}^{i} u_k+{w}^{i}_{k}) + G_{1} d_{1,k} +\hspace{-0.05cm} G_{2} d_{2,k}, 
\\  y_k
&={C} x_k + \textstyle\sum_{i=1}^{N} \lambda_{i,k} ({D}^{i} u_k +{v}^{i}_{k})+ H_{1} d_{1,k}    \label{eq:y}
\end{align}
where $G_{1} \triangleq G V_{1}$, $G_{2} \triangleq G V_{2}$ and $H_{1} \triangleq H V_{1}=U_{1} \Sigma$. Next, the output $y_k$ is decoupled 
using a nonsingular transformation $T =\begin{bmatrix} T_{1}^\top & T_{2}^\top \end{bmatrix}^\top \triangleq U^\top =\begin{bmatrix} U_{1} & U_{2} \end{bmatrix}^\top $ 
to get $z_{1,k} \in \mathbb{R}^{p_{H}}$ and $z_{2,k} \in \mathbb{R}^{l-p_{H}}$ given by
\begin{gather} \label{eq:sysY} \hspace{-0.2cm}\begin{array}{lll}
z_{1,k} &\triangleq T_{1} y_k =U_{1}^\top y_k \\ &= {C}_{1} x_k + \Sigma d_{1,k} + \sum_{i=1}^N \lambda_{i,k} {D}^{i}_1 u_k + \sum_{i=1}^N \lambda_{i,k}{v}^{i}_{1,k}\\
z_{2,k} &\triangleq T_{2} y_k =U_{2}^\top y_k \\ &=  {C}_{2}  x_k + \sum_{i=1}^N \lambda_{i,k} {D}^{i}_2 u_k + \sum_{i=1}^N \lambda_{i,k} {v}^{i}_{2,k}\vspace{-0.5cm}\end{array} 
\end{gather}
where ${C}_{1} \triangleq U_1^\top {C}$, ${C}_{2} \triangleq U_{2}^\top {C}$, ${D}^{i}_{1} \triangleq U_{1}^\top {D}^{i}$, ${D}^{i}_{2} \triangleq  U_{2}^\top {D}^{i}$, ${v}^{i}_{1,k} \triangleq U_{1}^\top {v}^{i}_k$ and ${v}^{i}_{2,k} \triangleq  U_{2}^\top {v}^{i}_k$. This transform is also chosen such that $\|\begin{bmatrix} {{v}^{i}_{1,k}}^\top & {{v}^{i}_{2,k}}^\top \end{bmatrix}^\top\|=\| U^\top {v}^{i}_k\|=\|{v}^{i}_k\|$. 
\section{Fixed-Order Simultaneous Input and State Set-Valued Observers} \label{sec:observer}
\subsection{Set-Valued Observer Design} \label{sec:obsv}
We consider a recursive three-step set-valued observer design. 
This design utilizes a similar framework as in 
 \cite{YongACC2018} and contains an \emph{unknown input estimation} step that uses the current measurement and the set of compatible states to estimate the set of compatible unknown inputs, 
a \emph{time update} step which propagates the compatible set of states based on the system dynamics, and a \emph{measurement update} step that uses the current measurement to update the  set of compatible states. To sum up, our target is to design a 
three-step recursive set-valued observer of the form:
\begin{align*}
\text{\emph{Unknown Input Estimation:}} & \ \hat{D}_{k-1} = \mathcal{F}_d(\hat{X}_{k-1},u_k),\\
\text{\emph{Time Update:}} & \quad \ \hat{X}^\star_k = \mathcal{F}_x^\star(\hat{X}_{k-1},\hat{D}_{k-1},u_k),\\
\text{\emph{Measurement Update:}} & \quad \ \hat{X}_k = \mathcal{F}_x(\hat{X}^\star_{k},u_k,y_k),
\end{align*}
where $\mathcal{F}_d$, $\mathcal{F}^\star_x$ and $\mathcal{F}_x$ are to-be-designed set mappings, while $\hat{D}_{k-1}$, $\hat{X}^\star_{k}$ and $\hat{X}_k$ are the sets of compatible unknown inputs at time $k-1$, propagated, and updated states at time $k$, correspondingly. It is important to note that $d_{2,k}$ cannot be estimated from $y_k$ since it does not affect $z_{1,k}$ and $z_{2,k}$.
Thus, the only estimate we can obtain in light of \eqref{eq:sysY} is a (one-step) delayed estimate of $\hat{D}_{k-1}$.
The reader may refer to a previous work \cite{Yong.Zhu.ea.CDC15_General} for a complete discussion on when a delay is absent or when we can expect further delays.
Similar to \cite{chen2005observer,blanchini2012convex,YongACC2018}, as the complexity of optimal 
observers increases with time, 
only the fixed-order recursive filters will be considered. In particular, we choose set-valued estimates of the form:
\begin{align*}
\begin{array}{rl}
\hat{D}_{k-1}&=\{d \in \mathbb{R}^p: \|d_{k-1}-\hat{d}_{k-1}\|\leq \delta^d_{k-1}\},\\
\hat{X}^\star_k&=\{x \in \mathbb{R}^n: \|x_k-\hat{x}^\star_{k|k}\|\leq \delta^{x,\star}_{k}\},\\
\hat{X}_k&=\{x \in \mathbb{R}^n: \|x_k-\hat{x}_{k|k}\|\leq \delta^x_k\}.
\end{array}
\end{align*}
In other words, we restrict the estimation errors to balls of norm $\delta$. In this setting, the observer design problem is equivalent to finding the centroids $\hat{d}_{k-1}$, $\hat{x}^\star_{k|k}$ and $\hat{x}_{k|k}$ as well as the radii $\delta^d_{k-1}$, $\delta_k^{x,\star}$ and $\delta_k^x$ of the sets $\hat{D}_{k-1}$, $\hat{X}^\star_{k}$ and $\hat{X}_k$, respectively. 
In addition, we limit our attention to observers for the centroids $\hat{d}_{k-1}$, $\hat{x}^\star_{k|k}$ and $\hat{x}_{k|k}$ that belong to the class of three-step recursive filters given in \cite{Gillijns.2007b} and \cite{Yong.Zhu.ea.Automatica16}, defined as follows for each time $k$ (with $\hat{x}_{0|0}=\hat{x}_0$):

\noindent \emph{Unknown Input Estimation}: \vspace{-0.1cm}
\begin{align}
&\hat{d}_{1,k} =M_{1} (z_{1,k}-{C}_{1}\hat{x}_{k|k} -\textstyle\sum_{i=1}^N\lambda_{i,k}{D}^{i}_{1} u_k), \label{eq:variant1}\\
&\hat{d}_{2,k-1}=M_{2} (z_{2,k}-{C}_{2}\hat{x}_{k|k-1} -\textstyle\sum_{i=1}^N\lambda_{i,k}{D}^{i}_{2} u_k), \label{eq:d2}\\
&\hat{d}_{k-1}= V_{1} \hat{d}_{1,k-1} + V_{2} \hat{d}_{2,k-1} \label{eq:d}. 
\end{align}
\emph{Time Update}: \vspace{-0.1cm}
\begin{align}
  \hat{x}_{k|k-1}\hspace{-0.1cm}&=\hspace{-0.1cm} \textstyle\sum_{i=1}^N\hspace{-0.1cm}\lambda_{i,k-1}\hspace{-0.05cm}(\hspace{-0.05cm} {A}^{i} \hat{x}_{k-1 | k-1} \hspace{-0.1cm}+\hspace{-0.1cm} {B}^{i} u_{k-1}\hspace{-0.05cm}) \hspace{-0.08cm} +\hspace{-0.08cm} G_{1}\hspace{-0.05cm} \hat{d}_{1,k-1},\label{eq:time} \\
\hat{x}^\star_{k|k}&=\hat{x}_{k|k-1}+G_{2} \hat{d}_{2,k-1}. \label{eq:xstar}
\end{align}
\emph{Measurement Update}: \vspace{-0.1cm}
\begin{align}
\begin{array}{rl}
 \hat{x}_{k|k}&=\hat{x}^\star_{k|k} + L(y_{k} - {C} \hat{x}^\star_{k|k}-\textstyle\sum_{i=1}^N\lambda_{i,k} {D}^{i} u_k) \\
&= \hat{x}^\star_{k|k} +\tilde{L}(z_{2,k}-{C}_{2} \hat{x}^\star_{k|k}-\textstyle\sum_{i=1}^N\lambda_{i,k} {D}^{i}_{2} u_k),   \hspace{-0.2cm} \label{eq:stateEst}
\end{array}
\end{align}
where $L \in \mathbb{R}^{n \times l}$, $\tilde{L} \triangleq L U_{2}\in \mathbb{R}^{n \times (l-p_{H})}$, $M_{1} \in \mathbb{R}^{p_{H} \times p_{H}}$ and $M_{2} \in \mathbb{R}^{(p-p_{H}) \times (l-p_{H})}$ are observer gain matrices that are designed according to Theorem \ref{thm:L_gain}. The main result in Theorem  \ref{thm:L_gain} is derived by minimizing the ``volume" of the set of compatible states and unknown inputs, quantified by the radii $\delta^d_{k-1}$, $\delta_k^{x,\star}$ and $\delta_k^x$.  
Note also that we applied $L=L U_{2} U_{2}^\top=\tilde{L} U_2^\top$ from Lemma \ref{lem:unbiased} 
into \eqref{eq:stateEst}.
 The state and input estimation errors are defined as $\tilde{x}_{k|k}\triangleq x_k- \hat{x}_{k|k},\tilde{d}_{k-1}\triangleq d_{k-1}-\hat{d}_{k-1}, \tilde{d}_{1,k-1}\triangleq d_{1,k-1}-\hat{d}_{1,k-1}, \tilde{d}_{2,k-1}\triangleq d_{2,k-1}-\hat{d}_{2,k-1}$ respectively.
In Lemmas \ref{lem:unbiased} and \ref{lem:uniform detecatbility}, we will provide necessary conditions for boundedness of estimation errors and sufficient conditions for stability of the observer.
All the proofs are provided in the Appendix.
\begin{lem}[Necessary Conditions for Boundedness of Set-Valued Estimates {\cite[Lemma 1]{YongACC2018}}] 
	\label{lem:unbiased}
	The input and state estimation errors, ($\tilde{d}_{k-1}$ and 
	$\tilde{x}_{k|k}$), are bounded for all $k$ (i.e., the set-valued estimates are bounded with radii $\delta^d_{k-1}, \delta_k^{x,\star}, \delta_k^x < \infty$), 
	only if 
	$ M_{1} \Sigma=I $, $p \leq l$,
	$ M_{2}C_{2} G_{2} = I $ and $ L U_{1}=0 $ . Consequently, ${\rm rk}(C_{2} G_{2})=p-p_{H}$, $M_1=\Sigma^{-1}$, $M_2=(C_2 G_2)^\dagger$ and $L=L U_{2} U_{2}^\top=\tilde{L} U_2^\top$.
\end{lem} 
\begin{lem}[Sufficient Conditions for Observer Stability] 
\label{lem:uniform detecatbility}
A sufficient condition for the stability of the set-valued observer is that $ (\overline{A}_k , C_{2}  ) $ is uniformly detectable\footnote{\ For conciseness, the readers are referred to \cite[Section 2]{Anderson.Moore.1981} for the definition of uniform detectability. A spectral test can be found in \cite{Peters.Iglesias.1999}.} for each $k$, where $ \overline{A}_k \triangleq ( I - G_2 M_2 C_2) \hat{A}_k $ and $ \hat{A}_k \triangleq \sum_{i=1}^N \lambda_{i,k} {A}^{i} - G_1 M_1 C_1 $. 
\end{lem}

\subsection{Optimal $\mathcal{H}_\infty$-Observer}\label{sec:Hinf}
In this section, we provide sufficient conditions for the \emph{existence} of a
set-valued observer for system \eqref{eq:mainSystem} with any sequence $\{\lambda_{i,k}\}_{k=0}^\infty$ for all $i \in \{1,2,\dots,N\}$ that satisfies $ 0 \leq \lambda_{i,k} \leq1 , \sum_{i=1}^N \lambda_{i,k} =1, \forall k $ in the sense of $\mathcal{H}_\infty$ (i.e., minimizing the sum of squares of the state estimation error sequence). Furthermore, we introduce a relatively simple approach to find such an observer, which involves solving a semi-definite program with Linear Matrix Inequalities (LMI) as constraints. We will also show that given some structural conditions for the system, the upper bounds of the estimation errors for both states and unknown inputs are guaranteed to converge to steady state.   
\begin{thm}[$\mathcal{H}_\infty$-Observer Design] 
\label{thm:L_gain}	
Suppose 
Lemma \ref{lem:unbiased} holds and there exist matrices $Y$ and $S \succ 0 $ with appropriate dimensions such that $$\begin{bmatrix} S & (\overline{{A}^{i}})^\top( S -  C_2^\top Y^\top )  & 0 &  I \\ * & S & \begin{bmatrix} S - Y C_2 & -Y \end{bmatrix} & 0 \\ * & * & \eta I & 0 \\ * & * & * & \eta I \end{bmatrix} \succ 0 $$ for all $ i \in\left\{1,2,\dots ,N\right\}.$ 
Then, there exists an $ \eta $ performance bounded $\mathcal{H}_\infty$-observer for system \eqref{eq:mainSystem} with any sequence $\{\lambda_{i,k}\}_{k=0}^\infty$ for all $i \in \{1,2,\dots,N\}$ that satisfies $ 0 \leq \lambda_{i,k} \leq1 , \sum_{i=1}^N \lambda_{i,k} =1, \forall k $ when using  $ \tilde{L} = {S}^{-1} Y $, i.e., $\| T_{\tilde{x},w,v} \| \leq \eta^2$, where  $T_{\tilde{x},w,v}$ is the transfer function matrix that maps the noise signals $\sum_{i=1}^N \lambda_{i,k} \begin{bmatrix}  w_{k}^{i \, \top} & v_k^{i \, \top} \end{bmatrix}^T$ to the updated state estimation error $\tilde{x}_{k|k}\triangleq x_k-\hat{x}_{k|k}$.

 Furthermore, the optimal filter gain $ \tilde{L} = {S^ \star} ^ {-1} \tilde{Y} ^ \star $ with $ \eta^\star $ $\mathcal{H}_\infty$-performance can be obtained from the following semi-definite programming with LMI constraints: 
 
\footnotesize \vspace{-0.35cm}
\begin{align}
\nonumber (\eta^\star,S^\star,Y^\star) &\in \argmin_{\eta, S,Y} \hspace{.2 cm} \eta
\\ \nonumber & \hspace*{-0.5cm} s.t \begin{bmatrix}  S & (\overline{{A}^{i}})^\top( S -  C_2^\top Y^\top )  & 0 & I \\ * & S & \begin{bmatrix} S - Y C_2 & -Y \end{bmatrix} & 0 \\ * & * & \eta I & 0 \\ * & * & * & \eta I \end{bmatrix} \succ 0,\\ & \hspace{3.75cm}\forall i \in\left\{1,2,.. ,N\right\}.
\label{eq:LMI}
\end{align}
\normalsize
\end{thm}
 Although Theorem \ref{thm:L_gain} equips us with an approach for designing an $\mathcal{H}_\infty$-observer for the linear parameter-varying system in \eqref{eq:mainSystem} when one exists, it would still be valuable to find a \emph{structural} and conveniently testable property for the constituent linear time-invariant systems in \eqref{eq:indSystem} that is \emph{necessary} for the existence of such an observer. Knowing such conditions would be beneficial in the sense that if they are \emph{not} satisfied, the designer knows \emph{a priori} that there does not exist any $\mathcal{H}_\infty$-observer for such an attacked system.  This will be the goal of Theorem \ref{thm:Hinfsufficient}.
 \begin{thm}[Necessary Conditions for the Existence of an $\mathcal{H}_\infty$-observer]
  \label{thm:Hinfsufficient}
  There exists a simultaneous state and unknown input $\mathcal{H}_\infty$-observer for system \eqref{eq:mainSystem} with any sequence $\{\lambda_{i,k}\}_{k=0}^\infty$ for all $i \in \{1,2,\dots,N\}$ that satisfies $ 0 \leq \lambda_{i,k} \leq1 , \sum_{i=1}^N \lambda_{i,k} =1, \forall k $, only if each $ ({A}^{i},G,{C},H) $ is strongly detectable\footnote{For brevity, the readers may refer to \cite{YongACC2018} for the definition of strong detectability.} for all $ i \in\left\{1,2,\dots ,N\right\} $.
 \end{thm}
 
 Next, we characterize the resulting radii $\delta^x_{k}$ and $\delta^d_{k-1}$ when using the proposed $\mathcal{H}_\infty$-observer.
 
 \begin{thm}[Radii of Set-Valued Estimates]\label{thm:error_bound}
	The radii $\delta^x_{k}$ and $\delta^d_{k-1}$ can be obtained as:
	
	\footnotesize \vspace{-0.35cm}
	\begin{align*}
	&\delta^x_{k}=\delta^x_0 \theta^k +  \overline{\eta} \textstyle\sum_{i=1}^k \theta^{i-1},\\ 
	 &\delta^d_{k-1} \hspace{-0.1cm}=\hspace{-0.1cm} \beta \delta^x_{k-1}\hspace{-0.1cm}+\hspace{-0.1cm} \| V_2 M_2 C_2 \| \eta_w
\hspace{-0.1cm} +\hspace{-0.1cm}\big{[}\hspace{-0.05cm}\|\hspace{-0.05cm}(\hspace{-0.05cm}V_2 M_2 C_2G_1\hspace{-0.1cm}-\hspace{-0.1cm}V_1\hspace{-0.05cm})\hspace{-0.05cm}M_1T_1\hspace{-0.05cm}\|\hspace{-0.1cm}+\hspace{-0.1cm}\|\hspace{-0.05cm}V_2M_2T_2\hspace{-0.05cm}\| \hspace{-0.05cm}\big{]} \eta_v,
	\end{align*}
	\normalsize
	where $\beta \triangleq \textstyle\max_{i\in \{1,2,\dots,N\}} \|V_1M_1C_1+ V_2 M_2 C_2 A_{e,i}\|,
	\Psi \triangleq I-\tilde{L}C_2, \Phi \triangleq I-G_2 M_2 C_2,  
	A_{e,i} \triangleq \Psi \Phi (A^i-G_1M_1C_1),
	 \\ \theta \triangleq \max_{i \in \{1,2,\dots ,N\}} \| A_{e,i} \|$.
\end{thm}

The resulting fixed-order set-valued observer is summarized in Algorithm \ref{algorithm1}.
 
So far, we have designed an $\mathcal{H}_\infty$-observer for our linear parameter-varying system and provided necessary conditions for the boundedness of the set-valued estimates.
 It is worth mentioning that 
 for the linear time-invariant case in \cite{YongACC2018},
   strong detectability of the system is also a sufficient condition for the convergence of the radii $\delta^x_{k}$ and $\delta^d_{k-1}$ to steady state. In our parameter-varying case, even if all constituent linear time-invariant systems are strongly detectable, there is no guarantee that the radii converge. The reason is that the convergence hinges on the stability of the product of \emph{time-varying} matrices (cf. proof of Theorem \ref{thm:error_convergence}), which is not guaranteed even if all the multiplicands are stable. 
 In the next theorem, 
 we discuss some sufficient conditions for the convergence of the radii to steady state.  
\begin{thm}[Convergence]
\label{thm:error_convergence}
	Suppose the conditions of Theorem \ref{thm:L_gain} hold. Then, the radii $\delta^x_{k}$ and $\delta^d_{k-1}$ are convergent if $\|A_{e,i}\|<1$ for all $i \in\left\{1,2,\dots ,N\right\}$, where $A_{e,i}$ is defined in Theorem \eqref{thm:error_bound}. Moreover, the steady state radii is given by:
	  
	  \small \vspace{-0.35cm}
 \begin{align*}
\begin{array}{l}
\displaystyle\lim_{k \to \infty} \delta^x_k  =\frac{ \overline{\eta}}{1-\theta}, \\
\displaystyle \lim_{k \to \infty} \delta^d_{k}=\frac{ \overline{\eta} \beta}{1-\theta}+\eta_w\| V_2 M_2 C_2 \| +\eta_v(\| V_2M_2T_2 \| + \| R \|),      
 \end{array}
\end{align*}
\normalsize
where 
$\overline{\eta} \triangleq (\| \Gamma \| \eta_v + \| \Psi \Phi \| \eta_w),   
R \triangleq V_2M_2C_2G_1M_1T_1-V_1 M_1 T_1, 
\Gamma \triangleq -(\Psi \Phi G_1 M_1 T_1 +\Psi G_2 M_2 T_2 + \tilde{L} T_2)$.  
\normalsize
\end{thm}
\begin{rem}
Alternatively, we can trade off between ``optimality" of the observer and ``convergence" of the radii. We can iteratively find $\eta$ (e.g., by line search) that satisfies the following feasibility problem:

\footnotesize \vspace{-0.35cm}
\begin{align*}
\begin{array}{rl}
\nonumber & {\rm Find} \hspace{.2 cm} (S,Y)
\\ & s.t \begin{bmatrix}  S & * & 0 & I \\ ( S - Y C_2 ) \overline{A}^{i} & S & \begin{bmatrix} S - Y C_2 & -Y \end{bmatrix} & 0 \\ * & * & \eta_0 I & 0 \\ * & * & * & \eta_0 I \end{bmatrix} \succ 0, \forall i \in\left\{1,2,.. ,N\right\},
\end{array}
\end{align*}
\normalsize 
as well as the sufficient condition in Theorem \ref{thm:error_convergence}, i.e., $\|A_{e,i}\|<1$ for all $i \in\left\{1,2,\dots ,N\right\}$.
%
%
 Although the designed observer may not be optimum in minimum $\mathcal{H}_\infty$ sense when using this alternative method, we can guarantee the steady state convergence of the radii instead. 
\end{rem}
\begin{algorithm}[t] \small
	\caption{Fixed-Order Input \& State Set-Valued Observer 
	}\label{algorithm1}
		\begin{algorithmic}[1]
		\State Initialize: $M_{1}=\Sigma^{-1}$; $M_{2}=(C_2 G_{2})^\dagger$; 
		\Statex \hspace{1.25cm} $\Phi=I-G_2 M_2 C_2$; 
		\Statex \hspace{1.25cm} 
		Compute $\tilde{L}$ via Theorem \ref{thm:L_gain}; 
		\Statex \hspace{1.25cm}  $\Psi=I-\tilde{L}C_2;$ 
		\Statex \hspace{1.25cm}   $\theta \triangleq \max_{i \in \{1,2,\dots ,N\}} \| \Psi \Phi (A^i-G_1M_1C_1) \|$; 
		\Statex \hspace{1.25cm} $\hat{x}_{0|0}=\hat{x}_0=\text{centroid}(\hat{X}_0)$; 
		\Statex \hspace{1.25cm} $\delta^x_0=\min\limits_{\delta} \{\|x-\hat{x}_{0|0}\| \leq \delta, \forall x \in \hat{X}_0\}$; 
		\Statex \hspace{1.25cm} $\hat{d}_{1,0}=M_1 (z_{1,0}-C_{1} \hat{x}_{0|0}-D_{1} u_0)$;
		\For {$k =1$ to $K$}
    	\LineComment{Estimation of $d_{2,k-1}$ and $d_{k-1}$}
		\State $\hat{x}_{k|k-1}= \sum_{i=1}^N\lambda_{i,k} {A}^{i} \hat{x}_{k-1 | k-1}+ \sum_{i=1}^N\lambda_{i,k} {B}^{i} u_{k-1}$ 
		\Statex \hspace{1.5cm}
		$ + G_{1} \hat{d}_{1,k-1};$
		\State  $\hat{d}_{2,k-1}=M_{2} (z_{2,k}-{C}_{2}\hat{x}_{k|k-1} -\sum_{i=1}^N\lambda_{i,k}{D}^{i}_{2} u_k);$
		\State $\hat{d}_{k-1} =V_{1} \hat{d}_{1,k-1} + V_{2} \hat{d}_{2,k-1}$; 
		\State $\delta^d_{k-1}=\delta_{k-1}^x\|V_1 M_1 C_1 + V_2 M_2 C_2 \hat{A}_k\| $
		\Statex \hspace{1.25cm}
		$+ \eta_v(\|(V_2 M_2 C_2G_1-V_1)M_1T_1\|+\|V_2M_2T_2\|)  $
		\Statex \hspace{1.25cm}
		$+ \eta_w\| V_2 M_2 C_2 \|;$
		\State $\hat{D}_{k-1}=\{d \in \mathbb{R}^l : \|d-\hat{d}_{k-1}\| \leq \delta^d_{k-1}\}$;
		\LineComment{Time update}
		\State
		$\hat{x}^\star_{k|k}=\hat{x}_{k|k-1}+G_{2} \hat{d}_{2,k-1}$;
		\LineComment{Measurement update}
		\State 
		$\hat{x}_{k|k}=\hat{x}^\star_{k|k}+\tilde{L}(z_{2,k}-C_{2} \hat{x}^\star_{k|k}- \sum_{i=1}^N \lambda_{i,k} D^i_{2} u_k)$;
		\State 
		$\delta^x_k=\delta^x_0 \theta^k +  \overline{\eta} \textstyle\sum_{i=1}^k \theta^{i-1}$; 
		\State  $\hat{X}_{k}=\{x \in \mathbb{R}^n : \|x-\hat{x}_{k|k}\| \leq \delta^x_{k}\}$;
		\LineComment{Estimation of $d_{1,k}$}
		\State $\hat{d}_{1,k}=M_{1} (z_{1,k}-C_{1} \hat{x}_{k|k}- \sum_{i=1}^N D^i_{1} u_k)$;
		\EndFor

	\end{algorithmic}
\end{algorithm}   

\vspace{-0.09cm}

\section{Simulation Example} \label{sec:examples}
In this section, we consider a convex combination of two constituent linear time-invariant strongly detectable subsystems that have been used in the literature as a benchmark for some state and input filters (e.g., \cite{chen2005observer}):

\small 
\vspace{-0.35cm}
\begin{align*}
&A^1 =\hspace{-0.05cm} \begin{bmatrix} 0.9 & .5 \\ -0.3 & 1 \end{bmatrix}\hspace{-0.1cm};
A^2 =\hspace{-0.05cm} \begin{bmatrix} 0.85 & .55 \\ -0.35 & 1 \end{bmatrix}\hspace{-0.1cm}; 
C=\hspace{-0.05cm} \begin{bmatrix} 1 & .2 \\ 1.1 & 1.9 \end{bmatrix}\hspace{-0.1cm};\\
&G =\hspace{-0.05cm} \begin{bmatrix} -0.02 & 0.04 \\ 0.01 & -0.05\end{bmatrix}\hspace{-0.1cm}; H=\hspace{-0.05cm}\begin{bmatrix} 1.1 & 2 \\ 2.2 & 4 \end{bmatrix}\hspace{-0.1cm};
B^1 =B^2= I_{2 \times 2};D = 0_{2 \times 2}.
\end{align*} 
\normalsize
The unknown inputs used in this example are as given in Figure \ref{fig:inputs}, while the initial state estimate and noise signals (drawn uniformly) have bounds $\delta^x_{0}=0.5$, $\eta_w=0.02$ and $\eta_v= 10^{-4}$. We also picked uniformly random coefficients, $\lambda_{i,k}$, that satisfies $ 0 \leq \lambda_{i,k} \leq1 , \sum_{i=1}^N \lambda_{i,k} =1, \forall k $.
Based on the results of Theorem \ref{thm:L_gain} and by solving the corresponding semi-definite programming problem using YALMIP \cite{Lofberg2004} and MOSEK \cite{mosek} as the solver, 
we find $S^\star=\begin{bmatrix} 0.2745 & 0.1933 \\
0.1933 & 0.4200 \end{bmatrix} $, $Y^\star= \begin{bmatrix} 0.0010 \\ 0.1613 \end{bmatrix} $ and the $\mathcal{H}_\infty$-observer gain as $ \tilde{L}=S^{\star \, -1}Y^\star= \begin{bmatrix} -0.3946 \\
0.5656 \end{bmatrix}$. Then, applying Algorithm \ref{algorithm1}, we summarized the set-valued state and unknown input results in Figures \ref{fig:inputs} and \ref{fig:variances}. The radii are observed to be convergent to steady state in Figure \ref{fig:variances}.     
\vspace{-0.1cm}
\begin{figure}[!h]
	\begin{center}
		\includegraphics[scale=0.385,trim=18mm 40mm 10mm 10mm,clip]{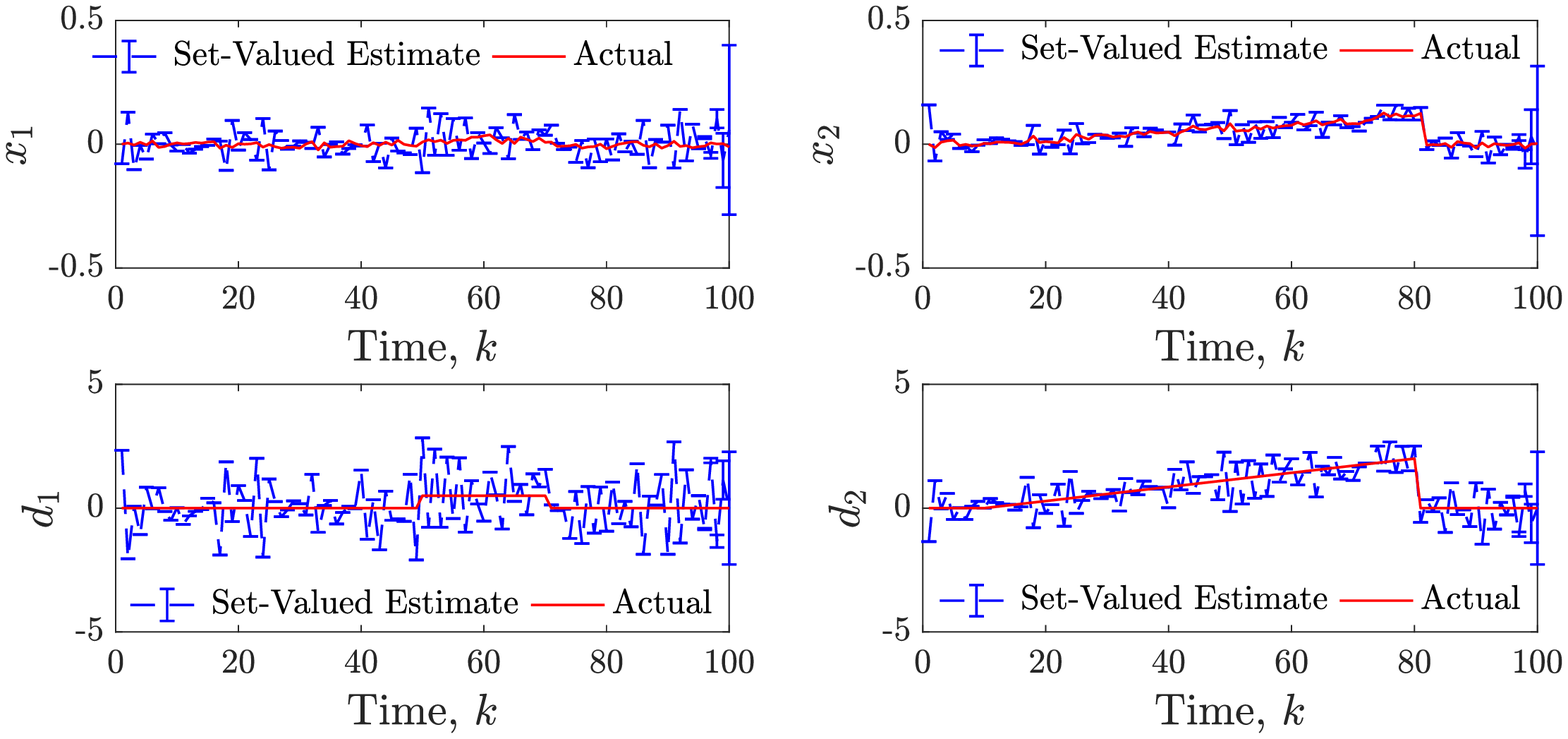}\vspace{-3.05cm}
		\caption{Actual states $x_1$, $x_2$ and their estimates, as well as unknown inputs $d_1$ and $d_2$ and their estimates.\label{fig:inputs} }
		\includegraphics[scale=0.37,trim=20mm 70mm 10mm 2mm,clip]{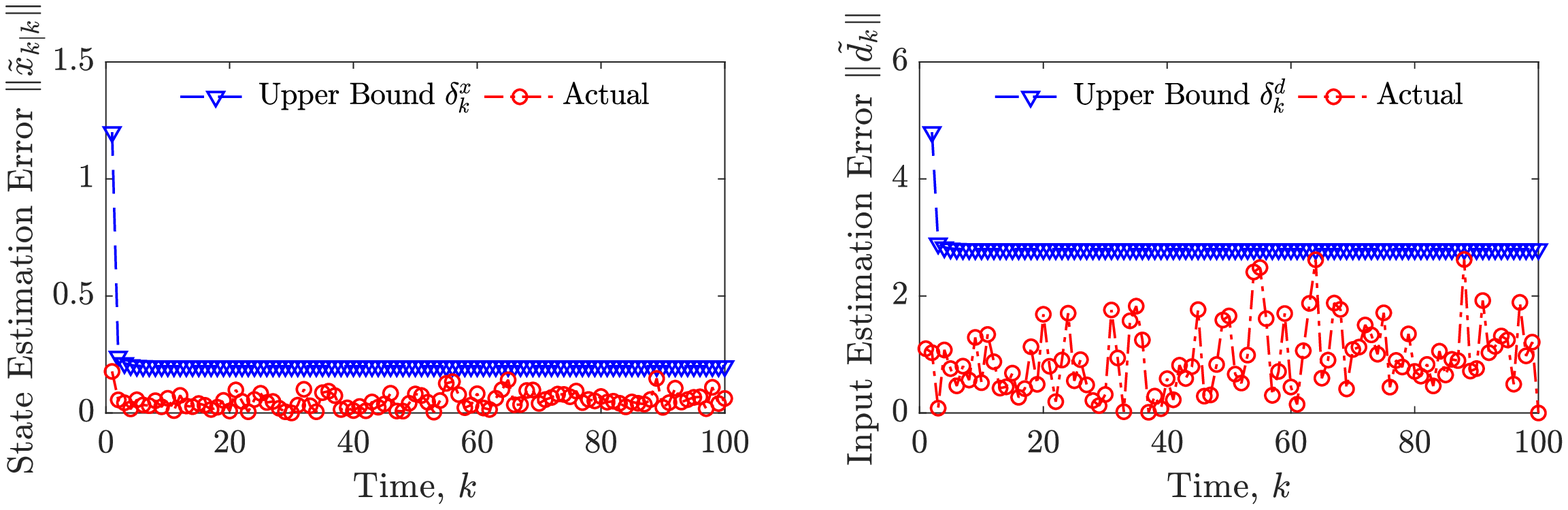}\vspace{-0.25cm}
		\caption{Actual estimation errors and radii of set-valued estimates of states, \hspace{-0.1cm}$\|\tilde{x}_{k|k}\|$, \hspace{-0.075cm}$\delta^x_k$, and unknown inputs, \hspace{-0.1cm}$\|\tilde{d}_{k}\|$, \hspace{-0.075cm}$\delta^d_k$. \hspace{-0.45cm} \label{fig:variances}}
	\end{center}
	\vspace{-0.425cm}
\end{figure}
\vspace{-0.1cm}
\section{Conclusion} \label{sec:conclusion}
\vspace{-0.05cm}
 We presented a fixed-order set-valued $\mathcal{H}_\infty$-observer for linear parameter-varying bounded-error discrete-time dynamic systems, which can be expressed as a convex combination of strongly detectable linear time-invariant constituent systems. 
  We provided sufficient conditions for the optimality of the designed observer, which can be obtained from a semi-definite programming problem with LMI constraints. We also showed that the strong detectability of the constituent linear time-invariant systems 
  is necessary 
  for the existence and stability of such an observer and for the boundedness of the set-valued estimates. In addition, we came up with sufficient structural conditions for the convergence of the radii of the set-valued state and input estimates and derived the steady state radii. Finally, we 
demonstrated the effectiveness of our proposed approach using an illustrative example.

\tiny
\bibliographystyle{unsrturl}
\bibliography{biblio}
\normalsize
\vspace{-0.3cm}
\section*{Appendix: Proofs} 
\label{subsec:thmproof} 
\subsection{Proof of Lemma \ref{lem:uniform detecatbility}}
 \eqref{eq:sysY}-\eqref{eq:time} and plugging  $ M_1=\Sigma^{-1} $ into \eqref{eq:estimatedinput1} imply that  
\begin{small}
\begin{align}
\hat{d}_{1,k}&= M_1 ( {C}_{1} \tilde{x}_{k|k} + \Sigma d_{1,k} + \textstyle{\sum}_{i=1}^N \lambda_{i,k} {v}^{i}_{1,k}), \label{eq:estimatedinput1}
\\ \hat{d}_{2,k-1}&\hspace{-0.1cm}=\hspace{-0.1cm} M_2 (C_2(\textstyle{\sum}_{i=1}^N \lambda_{i,k-1} A^i \tilde{x}_{k-1|k-1} \hspace{-0.1cm}+\hspace{-0.1cm} G_1 \tilde{d}_{1,k-1}\label{eq:estimatedinput2} \\ \nonumber &\ + G_2 d_{2,k-1} \hspace{-0.05cm}+\hspace{-0.1cm} \textstyle{\sum}_{i=1}^N \lambda_{i,k-1} w^i_{k-1}) \hspace{-0.05cm}+\hspace{-0.1cm} \textstyle{\sum}_{i=1}^N \lambda_{i,k} v^i_{2,k}). 
\\ \label{eq:inputerror1}  \tilde{d}_{1,k} &= d_{1,k} - \hat{d}_{1,k} = - M_{1} ( {C}_{1} \tilde{x}_{k|k} + \textstyle{\sum}_{i=1}^N \lambda_{i,k}{v}^{i}_{1,k} ).
\end{align}
\end{small}
 \eqref{eq:inputerror1} and setting $ M_{2}=( C_{2} G_{2} )^ \dagger $ (Lemma \ref{lem:unbiased}) in \eqref{eq:estimatedinput2}, 
return
 \small
 \begin{align} \label{eq:Inputerror2}
  \nonumber &\tilde{d}_{2,k-1}= -M_{2} ( C_{2} \hat{A}_{k-1} \tilde{x}_{k-1|k-1} - C_{2} G_{1} M_{1} \textstyle{\sum}_{i=1}^N \lambda_{i,k-1} {v}^{i}_{1,k-1} 
 \\&+\textstyle C_{2} \textstyle{\sum}_{i=1}^N \lambda_{i,k-1} {w}^{i}_{k-1} + \textstyle{\sum}_{i=1}^N \lambda_{i,k} {v}^{i}_{2,k}).  
\end{align}
\normalsize
Defining $ \tilde{x}^\star_{k|k} \triangleq x_k-\hat{x}^ \star_{k|k} $, from \eqref{eq:mainSystem}, \eqref{eq:time} and \eqref{eq:xstar} we obtain

\vspace{-0.3cm}
\begin{small}
\begin{align}
  \hspace{-0.05cm} \tilde{x}^\star_{k|k}\hspace{-0.1cm} =\hspace{-0.15cm} \textstyle{\sum}_{i=1}^N \hspace{-0.05cm}\lambda_{i,k-1}\hspace{-0.05cm} (\hspace{-0.05cm}{A}^{i} \tilde{x}_{k-1|k-1}\hspace{-0.1cm}+\hspace{-0.1cm}{w}^{i}_{k-1}\hspace{-0.05cm}) \hspace{-0.1cm}+\hspace{-0.1cm} G_{1}\hspace{-0.05cm} \tilde{d}_{1,k-1} 
   \hspace{-0.15cm}+\hspace{-0.1cm}\textstyle G_{2} \tilde{d}_{2,k-1}   \label{eq:staterrorstar} 
\end{align}
\end{small}

\vspace{-0.4cm}
In addition, from \eqref{eq:sysY} and \eqref{eq:stateEst} and \eqref{eq:inputerror1}-\eqref{eq:staterrorstar} we conclude:
\begin{align}
\tilde{x}_{k|k} 
&= ( I - \tilde{L} C_{2} ) \tilde{x}^\star_{k|k} - \tilde{L} \textstyle{\sum}_{i=1}^N \lambda_{i,k} {v}^{i}_{2,k}. \label{eq:staterror} 
\\ \tilde{x}^\star_{k|k}&= \overline{A}_{k-1} \tilde{x}_{k-1|k-1}-(I-G_2 M_2 C_2) (G_1 M_1  \label{eq:staterrorStar}\\ \nonumber &\textstyle\sum_{i=1}^N\lambda_{i,k-1} (v^i_{1,k-1}- w^i_{k-1})) -G_2 M_2 \sum_{i=1}^N \lambda_{i,k} v^i_{2,k}. 
\end{align}
 Now, we define
 \begin{small}
  $\overline{w}_{k-1} \hspace{-0.2cm}\triangleq -G_2 M_2 \textstyle \small{\sum_{i=1}^N} \lambda_{i,k} {v}^{i}_{2,k} - ( I - G_2 M_2 C_2 )
  ( G_1 M_1 \textstyle \small{\sum_{i=1}^N} \lambda_{i,k-1} ({v}^{i}_{1,k-1}-{w}^{i}_{k-1})) , 
     \overline{v}_{k-1} \triangleq \textstyle \sum_{i=1}^N \lambda_{i,k} {v}^{i}_{2,k}$. 
\end{small}
Then, \eqref{eq:staterror}-\eqref{eq:staterrorStar} imply that 
\small \vspace{-0.15cm}
\begin{align}
\begin{array}{rl}
\tilde{x}^\star_{k|k} &= \overline{A}_{k-1} \tilde{x}_{k-1|k-1} + \overline{w}_{k-1},  \label{eq:X_tilda_star}
\\  \tilde{x}_{k|k}\hspace{-0.05cm} &=\hspace{-0.05cm} ( I \hspace{-0.1cm}-\hspace{-0.1cm} \tilde{L} C_2 ) \overline{A}_{k-1} \tilde{x}_{k-1|k-1} \hspace{-0.1cm}+\hspace{-0.1cm} ( I \hspace{-0.1cm}-\hspace{-0.1cm} \tilde{L} C_2 ) \overline{w}_{k-1} \hspace{-0.1cm}-\hspace{-0.1cm} \tilde{L} \overline{v}_{k-1}. 
\end{array}
\end{align}
\normalsize
Now, consider the following linear time-varying system:
\begin{align} \label{eq:equivalentsys}
x_{k+1}=\overline{A}_k x_k + \overline{w}_k, y_k=C_2 x_k + \overline{v}_k.
\end{align}
Systems \eqref{eq:X_tilda_star} and \eqref{eq:equivalentsys} are equivalent from the viewpoint of estimation, since the estimation error equations for both problems are the same, hence they both have the same objective. Therefore, the pair $(\overline{A}_k,C_2)$ needs to be uniformly detectable such that the observer is stable \cite[Section 5]{Anderson.Moore.1981}.
\qed
\vspace{-0.15 cm}
\subsection{Proof of Theorem \ref{thm:L_gain}}
\vspace{-0.1cm}
Starting from \eqref{eq:X_tilda_star}, we have

\small \vspace{-0.3cm}
\begin{align} 
\nonumber  \tilde{x}_{k|k} = ( I - \tilde{L} C_2 ) \overline{A}_{k-1} \tilde{x}_{k-1|k-1} \hspace{-0.1cm}+\hspace{-0.1cm} ( I - \tilde{L} C_2 ) \overline{w}_{k-1} - \tilde{L} \overline{v}_{k-1},
\end{align}
\normalsize
from which we can define a system with $\tilde{x}_{k|k}$ as its state and $\tilde{z}_{k|k}=\tilde{x}_{k|k}$ as the output:
\vspace{-0.15cm}
\begin{align*}
\begin{array}{rl}
 \tilde{x}_{k|k} &= ( I - \tilde{L} C_2 ) \overline{A}_{k-1} \tilde{x}_{k-1|k-1} \hspace{-0.1cm}+\hspace{-0.1cm} \begin{bmatrix} I - \tilde{L} C_2 & -\tilde{L} \end{bmatrix}\hspace{-0.05cm} \begin{bmatrix} \overline{w}_{k-1} \\ \overline{v}_{k-1} \end{bmatrix}\hspace{-0.1cm},
\\  \tilde{z}_{k|k} &= \tilde{x}_{k|k}.
\end{array}
\end{align*}
By \cite[Lemma 3]{de2006robust}, this system
has an $\mathcal{H}_\infty$ performance bounded by $ \eta $, if there exists a symmetric positive definite matrix $ P $ with rank $ n $ such that:

\footnotesize \vspace{-0.2cm}
\begin{align}
\hspace{-0.15cm}\begin{bmatrix} P & ( I -  \tilde{L} C_2 ) \overline{{A}^{i}} P & \begin{bmatrix} I - \tilde{L} C_2 & - \tilde{L} \end{bmatrix} & 0 \\ * & P & 0 & P \\ * & * & \eta I & 0 \\ * & * & * & \eta I \end{bmatrix}\hspace{-0.1cm} \succ \hspace{-0.1cm}0 , \forall i \hspace{-0.1cm}\in\hspace{-0.1cm}\left\{\hspace{-0.05cm}1,2,\dots ,N\hspace{-0.05cm}\right\}.
\end{align}
\normalsize
Notice that the referenced lemma requires the existence of a \emph{bounded matrix sequence}, which in our case is a sequence of time-invariant matrices ($P$ is the same for each $k$), that is obviously bounded. 
By plugging $ S = {P}^{-1} \succ 0 $ and applying some similarity transformations, we obtain

\vspace{-0.4cm}
\begin{footnotesize} 
\begin{align*}&\begin{bmatrix} 0 & S & 0 & 0 \\ * & 0 & 0 & 0 \\ * & * & I & 0 \\ * & * & * & I \end{bmatrix} 
\hspace{-0.05cm}\begin{bmatrix} P & ( I -  \tilde{L} C_2 ) \overline{{A}^{i}} P & \begin{bmatrix} I - \tilde{L} C_2 & - \tilde{L} \end{bmatrix} & 0 \\ * & P & 0 & P \\ * & * & \eta I & 0 \\ * & * & * & \eta I \end{bmatrix}
\hspace{-0.05cm} \begin{bmatrix} 0 & S & 0 & 0 \\ * & 0 & 0 & 0 \\ * & * & I& 0 \\ * & * & * & I \end{bmatrix}\\ 
 &=\begin{bmatrix} S & \overline{{A}^{i}}^\top ( I - {C}_{2}^\top \tilde{L}^\top )S & 0 & I \\ * & S & \begin{bmatrix} I - \tilde{L}C_2 & - \tilde{L} \end{bmatrix} & 0 \\ * & * & I & 0 \\ * & * & * & \eta I \end{bmatrix} \hspace{-0.15cm} \succ \hspace{-0.1cm} 0 \ \forall i \hspace{-0.05cm}\in\hspace{-0.1cm}\left\{\hspace{-0.05cm}1,2,\dots ,N\hspace{-0.05cm}\right\}\hspace{-0.1cm}. \end{align*}
\end{footnotesize}

\vspace{-0.2cm}
Setting $ Y \triangleq S \tilde{L} $ completes the proof.
\qed
\vspace{-0.1cm}
\subsection{Proof of Theorem \ref{thm:Hinfsufficient}}
\vspace{-0.1cm}
Suppose, for contradiction, that there exists an $\mathcal{H}_\infty$-observer for system \eqref{eq:mainSystem} with any sequence $\{\lambda_{i,k}\}_{k=0}^\infty$ for all $i \in \{1,2,\dots,N\}$ that satisfies $ 0 \leq \lambda_{i,k} \leq1 , \sum_{i=1}^N \lambda_{i,k} =1, \forall k $, but one of the constituent linear time-invariant systems (e.g., $ ({A}^{j},G,{C},H) $) is not strongly detectable. Since the $\mathcal{H}_\infty$-observer exists for any sequence of $\lambda_{i,k}$, particularly it exists when $\lambda_{j,k}=1$ and $\lambda_{i j,k}=0$, $\forall i \neq j$ for all $k$. However, we know from \cite{YongACC2018} that strong detectability is necessary for the stability of the linear time-invariant system $ ({A}^{j},G,{C},H) $, which is a contradiction.
%
\qed 
\vspace{-0.1cm}
\subsection{Proof of Theorem \ref{thm:error_bound}}
\vspace{-0.1cm}
To prove Theorem \ref{thm:error_bound}, we first find closed form expressions for the state and input estimation errors in the following: 
\vspace{-0.2cm}
\begin{lem}\label{lem:error_closedform}
	The state and input estimation errors are 
	
	\footnotesize \vspace{-0.2cm}
	\begin{align*}
	 \tilde{x}_{k|k}\hspace{-0.05cm}&=\hspace{-0.05cm}(\textstyle\prod_{j=0}^{k-1} A_{e,k-j})\tilde{x}_{0|0} \hspace{-0.05cm}+\hspace{-0.05cm}\sum_{i=1}^k (\textstyle\prod_{j=0}^{i-2}A_{e,k-j})(\Psi \overline{w}_{k-i} \hspace{-0.05cm}-\hspace{-0.05cm} \tilde{L} \overline{v}_{k-i}),\\ 
\nonumber \tilde{d}_{k-1}\hspace{-0.05cm}&=\hspace{-0.05cm}\textstyle\sum_{i=1}^N \lambda_{i,k-1} (-V_1 M_1 C_1 - V_2 M_2 C_2 A_{e,i}) \tilde{x}_{k-1|k-1}   \\ \nonumber &+(-V_1M_1+V_2M_2C_2G_1M_1)T_1\textstyle\sum_{i=1}^N \lambda_{i,k-1} v^i_{k-1}
 \\ & -V_2 M_2 C_2 \textstyle\sum_{i=1}^N \lambda_{i,k-1} w^i_{k-1}-V_2M_2T_2 \sum_{i=1}^N \lambda_{i,k} v^i_k. 
	\end{align*}
	\normalsize
\end{lem}
\begin{proof}
From \eqref{eq:X_tilda_star}, we have
\begin{align} \label{eq:x_tilda_PSI}
\tilde{x}_{k|k} = \Psi \overline{A}_k \tilde{x}_{k-1|k-1} + \Psi \overline{w}_{k-1} - \tilde{L} \overline{v}_{k-1}. 
\end{align}
We use induction and \eqref{eq:x_tilda_PSI} to obtain

\footnotesize\vspace{-0.2cm}\balance
\begin{align*}
\begin{array}{rl}
 &\tilde{x}_{1|1} = \Psi \overline{A}_1 \tilde{x}_{0|0} + \Psi \overline{w}_{0} - \tilde{L} \overline{v}_{0}=A_{e,1} \overline{x}_{0|0}+ \Psi \overline{w}_{1-1} - \tilde{L} \overline{v}_{1-1} \\ 
&=( \textstyle\prod_{j=0}^{1-1} A_{e,1-1}) \tilde{x}_{0|0} +\textstyle\sum_{i=1}^1(\textstyle\prod_{j=0}^{i-2} A_{e,1-j})( \Psi \overline{w}_{k-1} -\tilde{L} \overline{v}_{1-i}) \\ 
& \tilde{x}_{k+1|k+1} = \Psi \overline{A}_{k+1} \tilde{x}_{k|k} +\Psi \overline{w}_k - \tilde{L} \overline{v}_k \\ 
&=\hspace{-0.1cm}\Psi \hspace{-0.05cm}\overline{A}_{k+1}\hspace{-0.05cm} \big{[}\hspace{-0.05cm}(\textstyle\prod_{j=0}^{k-1}\hspace{-0.1cm} A_{e,k-j})\tilde{x}_{0|0}\hspace{-0.05cm} %
   +\hspace{-0.15cm}\textstyle\sum_{i=1}^k \hspace{-0.05cm} (\textstyle\prod_{j=0}^{i-2}\hspace{-0.05cm}A_{e,k-j}\hspace{-0.05cm})\hspace{-0.05cm}(\hspace{-0.05cm}\Psi \overline{w}_{k-i}\hspace{-0.1cm} -\hspace{-0.1cm} \tilde{L} \overline{v}_{k-i}\hspace{-0.05cm})\hspace{-0.05cm}\big{]}
\\ &\quad+\Psi \overline{w}_k-\hspace{-0.1cm} \tilde{L} \overline{v}_k= (A_{e,k+1} A_{e,k}...A_{e,1}) \tilde{x}_{0|0}+\Psi \overline{w}_k - \tilde{L} \overline{v}_k \\ 
&\quad + \textstyle\sum_{i=1}^k (A_{e,k+1} A_{e,k}...A_{e,k-(i-2)})(\Psi \overline{w}_{k-i} - \tilde{L} \overline{v}_{k-i})) \\ 
&=\hspace{-0.05cm}(\textstyle\prod_{j=0}^{k+1} \hspace{-0.1cm}A_{e,k+1-j}) \tilde{x}_{0|0} \hspace{-0.05cm}+\hspace{-0.05cm}\textstyle\sum_{i=0}^k (\textstyle\prod_{j=0}^{i-2}A_{e,k-j})(\Psi \overline{w}_{k-i} \hspace{-0.05cm}-\hspace{-0.05cm} \tilde{L} \overline{v}_{k-i}) \\  &=\hspace{-0.05cm}(\textstyle\prod_{j=0}^{k+1}\hspace{-0.1cm} A_{e,k+1-j})\hspace{-0.05cm} \tilde{x}_{0|0} 
 \hspace{-0.15cm}+\hspace{-0.15cm}\textstyle\sum_{i=1}^{k+1}\hspace{-0.1cm} (\textstyle\prod_{j=0}^{i-2}\hspace{-0.1cm}A_{e,k+1-j}\hspace{-0.05cm})\hspace{-0.05cm}(\hspace{-0.05cm}\Psi \overline{w}_{k+1-i}\hspace{-0.1cm} -\hspace{-0.15cm} \tilde{L} \overline{v}_{k+1-i}\hspace{-0.05cm}).
\end{array}
\end{align*}\normalsize

As for $\tilde{d}_{k-1}$, \eqref{eq:inputerror1}-\eqref{eq:Inputerror2} imply
 
\small\vspace{-0.35cm}\begin{align}
\begin{array}{ll} 
 &\tilde{d}_{k-1}=V_1 \tilde{d}_{1,k-1}+V_2 \tilde{d}_{2,k-1}
\\ &=\textstyle\sum_{i=1}^N \lambda_{i,k-1} (-V_1 M_1 C_1 - V_2 M_2 C_2 A_{e,i}) \tilde{x}_{k-1|k-1}   \\ 
 &+(-V_1M_1+V_2M_2C_2G_1M_1)T_1\textstyle\sum_{i=1}^N \lambda_{i,k-1} v^i_{k-1}
 \\   &-V_2 M_2 C_2 \textstyle\sum_{i=1}^N \lambda_{i,k-1} w^i_{k-1}\hspace{-0.1cm}-\hspace{-0.1cm}V_2M_2T_2 \textstyle\sum_{i=1}^N \lambda_{i,k} v^i_k. \vspace{-0.4cm}\hspace{-0.25cm}  
 \end{array} \label{eq:d_tilda} 
 \end{align}\qedhere
\normalsize
\end{proof}

\vspace{-0.1cm}
Now, we are ready to prove Theorem \ref{thm:error_bound}. First, we define

\small \vspace{-0.35cm}
\begin{align}
\label{BekCektk}
  B_{e,k}\hspace{-0.05cm} \triangleq \hspace{-0.05cm}\textstyle\prod_{j=0}^{k-1} A_{e,k-j} , C^i_{e,k}\hspace{-0.05cm} \triangleq \hspace{-0.05cm} \textstyle\prod_{j=0}^{i-2} A_{e,k-j}, 
   \overline{t}_k \hspace{-0.05cm}\triangleq \hspace{-0.05cm} \Psi \overline{w}_k \hspace{-0.1cm}-\hspace{-0.1cm} \tilde{L} \overline{v}_k
    \end{align} \normalsize
   for $ 1 \leq i \leq k $. Then,
  from Lemma \ref{lem:error_closedform}, it follows that 
  \begin{small}
\begin{align}\label{eq:xtildanorm_inequal}
\begin{array}{rl}
 \| \tilde{x}_{k|k}\| &=\| B_{e,k} \tilde{x}_{0|0} + \textstyle\small{\sum_{i=1}^k} C^i_{e,k} \overline{t}_{k-i}  \|  \\
&\leq \| B_{e,k} \| \| \tilde{x}_{0|0} \| +\| \textstyle\sum_{i=1}^k C^i_{e,k} \overline{t}_{k-i}  \|. \end{array}
\end{align}
\end{small}

\vspace{-0.3cm}Moreover, by similar reasoning, 

\vspace{-0.3cm}
\begin{footnotesize}
\begin{align}
\nonumber &\| B_{e,k} \| =\| \textstyle\prod_{j=0}^{k-1} A_{e,k-j}\| \leq \textstyle\prod_{j=0}^{k-1} \| A_{e,k-j} \| \\
\nonumber &=\textstyle\prod_{j=0}^{k-1}\| \Psi \Phi \hat{A}_{k-j} \|=\textstyle\prod_{j=0}^{k-1}\| \Psi \Phi \textstyle\sum_{i=1}^N \lambda^i_{k-j} (A^i-G_1M_1C_1)\| 
\\ &= \textstyle\prod_{j=0}^{k-1}\|  \textstyle\sum_{i=1}^N \lambda^i_{k-j} \Psi \Phi (A^i-G_1M_1C_1) \| \label{eq:B_ek}
\\ \nonumber &\leq \textstyle\prod_{j=0}^{k-1} \textstyle\sum_{i=1}^N \lambda^i_{k-j} \| \Psi \Phi (A^i-G_1M_1C_1) \|\hspace{-0.05cm}\leq \hspace{-0.05cm} \textstyle\prod_{j=0}^{k-1} \theta=\theta^k,
\end{align}
\end{footnotesize}
\vspace{-0.35cm}
\begin{small}
\begin{align} \label{eq:Cekt}
 \|\hspace{-0.05cm} &\hspace{-0.05cm}\textstyle\sum_{i=1}^k\hspace{-0.05cm} C^i_{e,k} \overline{t}_{k-i}  \| \hspace{-0.05cm}\leq \hspace{-0.05cm} \textstyle\sum_{i=1}^k\hspace{-0.05cm} \| C^i_{e,k} \overline{t}_{k-i} \| \hspace{-0.1cm}\leq\hspace{-0.1cm} \textstyle\sum_{i=1}^k\hspace{-0.05cm} \|C^i_{e,k} \| \| \overline{t}_{k-i} \|,\hspace{-0.05cm}\\
 \nonumber &\| C^i_{e,k} \| = \| \textstyle\prod_{j=0}^{i-2} A_{e,k-j} \| \leq \textstyle\prod_{j=0}^{i-2} \|A_{e,k-j} \|
 \\ &=\textstyle\prod_{j=0}^{i-2} \| \textstyle\sum_{s=1}^N \lambda_{s,k-j} A_{e,s} \| \leq \textstyle\prod_{j=0}^{i-2} \theta \leq \theta^{i-1}. \label{eq:C_ek} 
\end{align}\end{small}
Furthermore, from the definition of $\overline{w}_{k}$ and \eqref{BekCektk} we have \\
\begin{small}
 $\overline{w}_{k-i} = - \Phi (G_1 M_1 \textstyle\sum_{s=1}^N \lambda_{s,k-i} v^s_{1,k-i} -\textstyle\sum_{s=1}^N \lambda_{s,k-i} w^s_{k-i}) 
  -G_2 M_2 \textstyle\sum_{s=1}^N \lambda_{s,k-i} v^s_{2,k-i}, 
 \| \overline{t}_{k-i}\| = \| \Psi \overline{w}_{k-i} - \tilde{L} \overline{v}_{k-i} \| =
  \| -\Psi \Phi G_1 M_1 T_1 \textstyle\sum_{s=1}^N \lambda_{s,k-i} v^s_{k-i} + \Psi \Phi \textstyle\sum_{s=1}^N \lambda_{s,k-i} w^s_{k-i} 
- \Psi  G_2 M_2 T_2 \textstyle\sum_{s=1}^N \lambda_{s,k-i} v^s_{k-i} - \tilde{L} T_2 \textstyle\sum_{s=1}^N v^s_{k-i} \|
= \|\textstyle\sum_{s=1}^N \lambda_{s,k-i} (\Gamma v^s_{k-i} +(\Psi \Phi) w^s_{k-i})\| \leq  \overline{\eta}$, 
\end{small}
 from which, as well as \eqref{eq:xtildanorm_inequal}-\eqref{eq:C_ek},
we conclude that

\vspace{-0.4cm}
\begin{small}
\begin{align} 
 \| \tilde{x}_{k|k} \| \hspace{-0.05cm}\leq\hspace{-0.05cm} \| \tilde{x}_{0|0} \| \theta^k \hspace{-0.05cm}+\hspace{-0.05cm}  \overline{\eta} \textstyle\sum_{i=1}^k \theta^{i-1}\hspace{-0.05cm}=\hspace{-0.05cm} \| \tilde{x}_{0|0} \| \theta^k \hspace{-0.05cm}+\hspace{-0.05cm} \overline{\eta} \frac{1-\theta^k}{1-\theta}\hspace{-0.05cm}\triangleq \hspace{-0.05cm}\delta^x_k.\hspace{-0.05cm}  \label{eq:Staterror}
  \end{align} 
  \end{small}

\vspace{-0.4cm}   As for $\delta^d_{k-1}$, using Lemma \ref{lem:error_closedform} and \eqref{eq:d_tilda}, triangle inequality and the facts that $ 0 \leq \lambda_{i,k} \leq 1, \textstyle\sum_{i=1}^N \lambda_{i,k}=1 $ and submultiplicativity of matrix norms, we obtain the result. \qed

\vspace{-0.05cm}
\subsection{Proof of Theorem \ref{thm:error_convergence}}

 Notice that
$0 \leq\|A_{e,i}\|\leq\theta<1$ for all $i \in\left\{1,2,\dots ,N\right\}$ by assumption. So, $\theta^k$ 
in \eqref{eq:Staterror} 
vanishes in steady state,
which gives us the following steady state estimation radius:
\begin{small}
$\lim_{k \to \infty} \delta^x_k =\lim_{k \to \infty} \left(\| \tilde{x}_{0|0} \| \theta^k + \overline{\eta} \frac{1-\theta^k}{1-\theta}\right)=\frac{ \overline{\eta}}{1-\theta}$. 
 \end{small}
Using this and starting from the expression for $\delta^d_{k-1}$ in Theorem \ref{thm:error_bound}, it converges to steady state, as follows:
\begin{small}
$\lim_{k\to \infty} \delta^d_{k-1}
 = (\lim_{k\to \infty} \beta \delta^x_{k-1})\hspace{-0.05cm}+\hspace{-0.05cm} \| V_2 M_2 C_2 \| \eta_w
   +(\|(V_2 M_2 C_2G_1-V_1)M_1T_1\|+\|V_2M_2T_2\|) \eta_v
 =\frac{ \overline{\eta} \beta}{1-\theta}+\eta_w\| V_2 M_2 C_2 \| +\eta_v(\| V_2M_2T_2 \| + \| R \|)$.\qed 
 \end{small}

\end{document}